\documentclass[submission,copyright,creativecommons]{eptcs}
\providecommand{\event}{LSFA 2026} 
\usepackage{iftex}

\usepackage{amsmath}
\usepackage{amsthm}
\usepackage{bussproofs}
\usepackage{multicol}
\usepackage{graphicx}
\usepackage{tikz-cd}
\usepackage[capitalise]{cleveref}

\usepackage{mymacros}

\title{Linearising Explicit Substitutions using Intersection Types}
\author{Ana Jorge Almeida 
\institute{LIACC}
\institute{Faculdade de Ciências,\\
Universidade do Porto}
\email{anaj.almeida@fc.up.pt}
\and
Sandra Alves
\institute{CRACS, INESC-TEC}
\institute{Faculdade de Ciências,\\
Universidade do Porto}
\email{sandra@fc.up.pt}
\and
Mário Florido
\institute{LIACC}
\institute{Faculdade de Ciências,\\
Universidade do Porto}
\email{amflorid@fc.up.pt}
}
\def\titlerunning{Linearising Explicit Substitutions using Intersection Types}
\def\authorrunning{A. J. Almeida, S. Alves \& M. Florido}

\begin{document}
\maketitle

\begin{abstract}
    Term expansion was originally introduced in 2004 as a way to relate terms typed in an intersection type system with linear terms. Recently, new applications of term expansion include the relation of lambda-terms with terms typed in other substructural type systems, such as the relevant and the ordered type systems, and the use of quantitative types to relate the strongly normalising lambda-terms with weak linear terms that share the same normal form. Here we define a new term expansion for a calculus with explicit substitutions, using it to relate a $\lambda$-calculus with explicit substitutions to Boudol's resource aware $\lambda$-calculus with multiplicities, where function arguments have a possibly limited availability. 
\end{abstract}

 \section{Introduction}

{\em Term expansion} was first defined in \cite{florido2004linearization} to relate $\lambda$-terms typed in an intersection type system with the linear $\lambda$-calculus. Recently, new applications of {\em term expansion} include the relation with other substructural type systems (relevant and ordered type systems) \cite{alves2022structural}  and the use of quantitative types \cite{alves2022quantitative}. Term expansion is related with several other  works on linearisation of the $\lambda$-calculus \cite{alves2005weak,ehrhard2008uniformity,kfoury2000linearization,mazza2017polyadic}, and contributes to this line of research by providing a simple uniform framework for addressing linearisation related problems. 

Although the $\lambda$-calculus is a convenient model for computational functions, it lacks the means for observing operational properties of the execution of such algorithms, mainly due to its implicit $\beta$-contraction, which is a meta-operation. Over the years, there was a necessity to explicitly deal with substitutions, in order to bridge the gap between theory and implementation, allow efficient reduction in implementations, and avoid variable capture and scope issues \cite{abadi1989explicit}. 

None of the previous works on term expansion and linearisation addresses calculi with explicit substitutions. In this paper we fill this gap by defining a new {\em term expansion} for a calculus with explicit substitutions \cite{rose1996explicit}. We explore its expressiveness by relating it with a resource aware calculus \cite{boudol1993lambda}, where the argument of a function is a bag of either unlimited or limited resources, that is, a multiset of terms. 

We use a modification of the explicit substitution calculus presented in \cite{rose1996explicit}, the $\lambda{\text{xgc}}$-calculus, which is an adaptation of $\lambda\sigma$ \cite{abadi1989explicit} that retains variable names instead of using indices \textit{à la De Bruijn} \cite{de1972lambda}, while preserving strong-normalisation. The other calculus we use is the $\lambda$-calculus with multiplicities \cite{boudol1993lambda} by Gérard Boudol, in which applications  $MN^m$ and substitutions $\boudolsubst{M}{N^m}{x}$, carry an explicit multiplicity $m$, meaning that term $N$ is of possibly limited availability.

Let us now look at the term $t \equiv \expsubst{(xx)}{x}{I}$ of the $\lambda$-calculus with explicit substitutions, with type $\alpha \to \alpha$, where $I$ is the identity function. Its {\em expansion} is the term $\boudolsubst{(x\finboudol{x}{1})}{\finboudol{I}{2}}{x}$ in Boudol's $\lambda$-calculus with multiplicities \cite{boudol1993lambda}, meaning that the identity function $I$ may be copied twice during reduction, in this case, corresponding to the number of free occurrences of $x$ in $xx$. Note that the expanded term has the same type as the original one, $\alpha \to \alpha$. The {\em expansion} of a term can thus be interpreted as a different version of the original term, where multiplicities are explicitly included in the term syntax. In the expanded version, we are able to meticulously track the resource aware behaviour of the term, providing us a better understanding of how resources are being used throughout term evaluation.

The paper is organised as follows. Section 2 presents the $\lambda$-calculus with explicit substitutions, including our proposed modifications. Section 3 presents Boudol's $\lambda$-calculus with multiplicities and introduces some auxiliary lemmas we will need to use further on. Section 4 presents two definitions of term expansion: we will first use a notion of expansion based on associative, commutative and idempotent intersection types, and then a notion of expansion based on associative, commutative and non-idempotent intersection types, and use each to show how to relate terms with explicit substitutions to terms in Boudol's $\lambda$-calculus with multiplicities.
\section{Explicit substitutions} 

The classical presentation of the $\lambda$-calculus \cite{barendregt1984lambda} lacks the means for observing the operational properties of the execution of computational functions, mainly because the process of substitution is a meta-operation \cite{abadi1989explicit}. A consequence of this is that the implementations have to deal with substitutions in a way that distances them from the theoretical calculus. Therefore, it became necessary to find new approaches that deal with substitutions explicitly, in order to bridge this gap between theory and implementation, allow efficient reduction in implementations, and avoid variable capture and scope issues.

One proposed solution is the $\lambda\sigma$-calculus \cite{abadi1989explicit}, developed as a way to design, understand, verify and compare implementations of the $\lambda$-calculus, from interpreters to machines. This calculus relies on De Bruijn's notation \cite{de1972lambda}, using indices instead of variable names, for example, $\lambda xy. xy$ becomes $\lambda \lambda 2 1$.
Although this notation leads to simple formal systems and solves the issue of variable capture, it also means that, if we build more complex $\lambda$-terms, its notation becomes unreadable, and thus harder to study.

The $\lambda$xgc-calculus \cite{rose1996explicit} emerged in an attempt to escape De Bruijn's indices, by proposing a calculus of explicit substitutions in the tradition of the $\lambda\sigma$-calculus, but retaining variable names. Here, explicit substitution is given highest precedence, and it has explicit garbage collection \cite{rose1992explicit}, which is useful and easy to specify using names.

\begin{definition}[$\lambda$x-preterms] \label[definition]{expsubst:preterms}
    The $\lambda$x-preterms are the extension of the $\lambda$-preterms defined inductively by
    \[ M ::= x \mid \lambda x.M \mid MN \mid \expsubst{M}{x}{N} \]
\end{definition}
\begin{definition}[ $\lambda$x-terms] \label[definition]{expsubst:terms}The $\lambda$x-terms are defined from  $\lambda$x-preterms modulo $\alpha$-equivalence. The notions of free variables, renaming and $\alpha$-equivalence are defined as expected.
    \begin{enumerate}
        \item The free variable set of a $\lambda$x-preterm $M$ is denoted \textit{fv}($M$) and defined inductively over $M$ by:
            \[    
            \begin{array}{rcll}
                \textit{fv}(x) &=& \{x\} \\
                \textit{fv}(\lambda x.M) &=& \textit{fv}(M) \setminus \{x\} \\
                \textit{fv}(MN) &=& \textit{fv}(M) \cup \textit{fv}(N) \\
                \textit{fv}(\expsubst{M}{x}{N}) &=& (\textit{fv}(M) \setminus \{x\}) \cup \textit{fv}(N)
            \end{array} 
            \]
            A $\lambda$x-preterm $M$ is closed if and only if \textit{fv}($M$) = $\emptyset$.
        \item The result of renaming all free occurrences of $y$ in $M$ to $z$ is written $M[y := z]$ and defined inductively over $M$ by:
            \[
            \begin{array}{rcll}
                x[y := z] &=& z &\text{if } x = y \\
                x[y := z] &=& x &\text{if } x \neq y \\
                (\lambda x.M)[y := z] &=& \lambda x'. M[x := x'][y := z] &\text{with } x' \notin \textit{fv}(\lambda x. M) \cup \{y, z\} \\
                (MN)[y := z] &=& (M[y := z])(N[y := z]) \\
                (\expsubst{M}{x}{N})[y := z] &=& \expsubst{M[x := x'][y := z]}{x'}{N[y := z]} &\text{with } x' \notin \textit{fv}(\lambda x. M) \cup \{y, z\}
            \end{array}
            \]
        \item That two terms are $\alpha$-equivalent is written $M \equiv N$. This means that they are identical except for renaming of bound variables, which is defined inductively by:
            \[
            \begin{array}{rcll}
                x &\equiv& x \\
                \lambda x. M &\equiv& \lambda y. N &\text{if } M[x := z] \equiv N[y := z] \text{ for } z \notin \textit{fv}(MN) \\
                MN &\equiv& PQ &\text{if } M \equiv P \text{ and } N \equiv Q \\
                \expsubst{M}{x}{N} &\equiv& \expsubst{P}{y}{Q} &\text{if } N \equiv Q \text{ and } M[x := z] \equiv P[y := z] \text{ for } z \notin \textit{fv}(MP)
            \end{array}
            \]
        \item The set of $\lambda$x-terms, $\Lambda$x, is the set of $\lambda$x-preterms modulo $\equiv$. The set of closed $\lambda$x-terms, $\Lambda$x$^\circ$, is the subset of $\Lambda$ where the representatives are closed preterms.
    \end{enumerate}
\end{definition}
\begin{definition}[$\lambda$xgc-reduction] Define the following reductions on $\lambda$x-terms. \label[definition]{expsubst:reduction}
    \begin{enumerate}
        \item \textit{substitution generation}, $\barrow$, is the contextual closure of $(\lambda x. M)N \to \expsubst{M}{x}{N}$.
        \item \textit{explicit substitution}, $\xarrow$, is defined as the contextual closure of the union of 
        \[ 
        \begin{array}{rcll}
             \expsubst{x}{x}{N} &\xvarrow& x \\
             \expsubst{x}{y}{N} &\xvgcarrow& x &\text{if } x\not\equiv y \\
             \expsubst{(\lambda x.M)}{y}{N} &\xabarrow& \lambda x. \expsubst{M}{y}{N} \\
             \expsubst{(M_1M_2)}{y}{N} &\xaparrow& \expsubst{M_1}{y}{N} \expsubst{M_2}{y}{N}
        \end{array}
        \]
        \item \textit{garbage collection}, $\gcarrow$, is the contextual closure of $\expsubst{M}{x}{N} \to M$ if $x \notin \textit{fv}(M)$. The subterm $N$ in $\expsubst{M}{x}{N}$ is called \textit{garbage} if $x \notin \textit{fv}(M)$.
        \item $\lambda$xgc-reduction is $\singlearrow = \barrow \cup \xarrow \cup \gcarrow$, $\lambda$x-reduction is $\bxarrow = \barrow \cup \xarrow$ and $\xgcarrow = \xarrow \cup \gcarrow$.
    \end{enumerate}
\end{definition}

\cref{expsubst:orig_ex} serves as a simple example intended to showcase how reductions are performed in the $\lambda$xgc-calculus. The evaluation rules will be explored more thoroughly later on.

\begin{example}~\label[example]{expsubst:orig_ex}
    Consider the $\lambda{\text{x}}$-term $(\lambda x. xx)I$, where $I \equiv \lambda z.z$.
    \[
    \begin{array}{rcll}
        (\lambda x. xx)I &\barrow& \expsubst{(xx)}{x}{I} \\
        &\xaparrow& \expsubst{x}{x}{I} (\expsubst{x}{x}{I})  \xvarrow I(\expsubst{x}{x}{I}) \xvarrow II \barrow \expsubst{z}{z}{I} \xvarrow I
    \end{array}
    \]
\end{example}

However, we wish to emulate the behaviour of functional programming languages, so we suggest some alterations that approximate the $\lambda$xgc-calculus reduction strategy to that of mainstream programming languages.

\subsection{Weak-head reduction}

Functional language compilers consider only weak-head reductions, in which $\beta$-reduction is performed only at the outermost, leftmost redex, and reductions do not occur under $\lambda$-abstractions \cite{fradet1994compilation}. These reductions stop evaluation when a weak-head normal form is produced, i.e., terms where no further reduction is possible at the head position.

\begin{figure}[h!]
    \figline
    \[ \expsubst{(M_1 M_2)}{x}{N} \xaparrow \expsubst{M_1}{x}{N} \expsubst{M_2}{x}{N} \ \text{if } x \in \textit{fv}(M_1) \text{ and } x \in \textit{fv}(M_2) \]
    \[ \expsubst{(M_1 M_2)}{x}{N} \xaparrow \expsubst{M_1}{x}{N} M_2 \ \text{if } x \in \textit{fv}(M_1) \text{ and } x \notin \textit{fv}(M_2) \]
    \[ \expsubst{(M_1 M_2)}{x}{N} \xaparrow M_1 \expsubst{M_2}{x}{N} \ \text{if } x \notin \textit{fv}(M_1) \text{ and } x \in \textit{fv}(M_2) \]
    \begin{multicols}{2}
        \prooftree
        \AxiomC{$M \ctxarrow M'$}
        \UnaryInfC{$MN \ctxarrow M'N$}
        \endprooftree
    
        \prooftree
        \AxiomC{$MN_2 \ctxarrow M' \ \ \ \text{if} \ x \notin \textit{fv}(N_2)$}
        \UnaryInfC{$(\expsubst{M}{x}{N_1})N_2 \ctxarrow \expsubst{M'}{x}{N_1}$} 
        \endprooftree
    \end{multicols}
    \figline
\caption{Evaluation in the $\lambda$xgc-calculus}
\label{expsubst:evaluation}
\end{figure}

\cref{expsubst:evaluation} introduces our proposed modification to the evaluation rules presented by Rose \cite{rose1996explicit}. We maintain reductions $\barrow$, $\xvarrow$, $\xvgcarrow$ and $\gcarrow$ from \cref{expsubst:reduction}, and remove evaluations under $\lambda$-abstractions, rule $\xabarrow$; perform evaluations only on the left side of applications; and introduce a new rule where reduction is performed at a distance -- this rule allows us to ``jump'' over explicit substitutions appended to $\lambda$-abstractions, in order to reach the application argument and reduce the term further \cite{accattoli2010structural}. These last two rules are contextual rules, and we have defined them as $\ctxarrow \in \{\singlearrow\}$. We also modified the $\xaparrow$ rule where, instead of always distributing the substitution to both parts of the application, we need to check where that variable occurs, allowing us to guarantee that we will not distribute substitutions where a variable does not occur. This is helpful since we are dealing with weak-head reductions, which means that sometimes we are not able to reach a function argument until much later in the evaluation, hence it is useful to have control of the propagation of substitution. For simplicity, we shall often refer to evaluations in this calculus using the general arrow $\singlearrow$.

\begin{example} \label[example]{expsubst:mod_orig_ex}
    Consider the $\lambda$x-term $\Delta I$, where $\Delta \equiv \lambda x. xx$ and $I \equiv \lambda z.z$.
    \[
    \begin{array}{rcll}
        \Delta I \barrow \expsubst{(xx)}{x}{I} &\xaparrow& \expsubst{x}{x}{I} (\expsubst{x}{x}{I}) \\
        &\xvarrow& I (\expsubst{x}{x}{I}) \barrow \expsubst{z}{z}{(\expsubst{x}{x}{I})} \xvarrow \expsubst{x}{x}{I} \xvarrow I
    \end{array}
    \]
\end{example}

Following \cref{expsubst:orig_ex}, the proposed changes can be seen in \cref{expsubst:mod_orig_ex}. 

These modifications also mean that the resulting terms will be in closed weak-head normal form, which are $\lambda$-abstractions or $\lambda$-abstractions with explicit substitutions.

\begin{example} \label[example]{expsubst:mod_compare_ex}
    Consider the $\lambda$x-term $\expsubst{(yI)}{y}{Z}$, where $I \equiv \lambda x.x$ and $Z \equiv \lambda z_1 z_2. z_1 z_2$. 
    
    We shall, firstly, reduce it using the $\xaparrow$ reduction as defined in \cref{expsubst:reduction}. 
    \[
    \begin{array}{rcll}
        \expsubst{(yI)}{y}{Z} &\xaparrow& \expsubst{y}{y}{Z}(\expsubst{I}{y}{Z}) \xvarrow Z(\expsubst{I}{y}{Z}) \barrow \expsubst{(\lambda z_2.z_1z_2)}{z_1}{(\expsubst{I}{y}{Z})} 
    \end{array}
    \]
    Now, we evaluate the term considering the $\xaparrow$ reduction as defined in \cref{expsubst:evaluation}.
    \[
    \begin{array}{rcll}
        \expsubst{(yI)}{y}{Z} &\xaparrow& \expsubst{y}{y}{Z}I \xvarrow ZI \barrow \expsubst{(\lambda z_2.z_1z_2)}{z_1}{I} 
    \end{array}
    \]
\end{example}
\section{The \texorpdfstring{$\lambda$}{lambda}-calculus with multiplicities}

The $\lambda$-calculus does not originally track resources, meaning that each argument is infinitely available. Looking at the $\beta$-reduction rule $(\lambda x. M)N \to_\beta \origsubst{M}{N}{x}$, it is intuitive that the argument $N$ is available however many times $x$ occurs free in $M$.  

Motivated by the study of the encoding of the lazy $\lambda$-calculus into the $\pi$-calculus \cite{milner1992functions}, Boudol \cite{boudol1993lambda} proposed a refinement of the $\lambda$-calculus, where function arguments are bags of resources with multiplicities that indicate how many copies of them are available for usage. The previous reduction is then altered to $((\lambda x. M)\infboudol{N}) \boudolarrow \boudolsubst{M}{\infboudol{N}}{x}$, explicitly indicating how many times argument $N$ is available.

\subsection{Syntax}

\begin{definition}[Boudol's terms] \label[definition]{boudol:terms}
The following terms are the syntax of Boudol's $\lambda$-calculus with multiplicities.
\[ 
\begin{array}{rcll}
    M &::=& x \mid \lambda x.M \mid (MP) \mid \boudolsubst{M}{P}{x} &\text{terms}\\
    P &::=& 1 \mid M \mid (P \mid P) \mid \finboudol{M}{m} &\text{bags of terms}\\
    V &::=& \lambda x. M \mid \boudolsubst{V}{P}{x} &\text{values}
\end{array}
\]
\end{definition} 
\cref{boudol:terms} shows the syntax of the $\lambda$-calculus with multiplicities \cite{boudol1993lambda}. Our terms are standard (variables, abstractions, applications and explicit substitutions), the bags of terms may be 1 (the neutral element), a term $M$, a parallel composition of bags of terms, which is commutative and associative, or a term available $m$ times -- we will look into this later. Finally, we have values, which are terms in weak-head normal form. 

Now, consider the term $(M\finboudol{N}{m})$, where $m \in \mathbb{N} \cup \{\infty\}$, and $\finboudol{N}{m}$ can be written as a parallel composition $\finboudol{N}{m} = (N \mid \cdots \mid N)$, $m$ times. Similarly, if $m = \infty$, it means that $\infboudol{N} = (N \mid N \mid \cdots)$, hence $\infboudol{N}$ is an infinite parallel composition of copies of $N$. The parallel composition is defined as follows \cite{boudol1993lambda}.
\begin{definition} Parallel composition is intended to be commutative and associative, with 1 as its neutral element. \label{boudol:parallel_comp}
    \[
    \begin{array}{rcll}
        (P \mid 1) &\equiv& P \\
        (P \mid Q) &\equiv& (Q \mid P) \\
        (P \mid (Q \mid R)) &\equiv& ((P \mid Q) \mid R) \\
        \finboudol{M}{0} &\equiv& 1 \\
        \finboudol{M}{m+1} &\equiv& (M \mid \finboudol{M}{m})
    \end{array}
    \]
\end{definition}

Boudol \cite{boudol1993lambda} proposes a general set $\Lambda^m$, where $m = \mathbb{N} \cup \{\infty\}$. However, from $\Lambda^m$, we will distinguish the subset that deals with only infinite multiplicities, $\Lambda^\infty$, and the subset that deals only with finite multiplicities, which we will refer to as $\Lambda$.

Dealing with finite multiplicities, $m \in \mathbb{N}$, suggests that the argument may not always be available for usage. Consequently, there is the possibility of \textit{deadlock} when we have fewer resources available than variables to be substituted. A deadlocked term is a term that cannot be reduced any further but is not a value.

\subsection{Evaluation}
\begin{figure}[h!] 
    \figline
    
    Evaluation follows a weak-head reduction strategy.

    \begin{multicols}{2}
    \prooftree
    \AxiomC{$M \boudolarrow M'$}
    \RightLabel{\scriptsize E$_1$}
    \UnaryInfC{$MP \boudolarrow M'P$}
    \endprooftree

    \prooftree
    \AxiomC{$M \boudolarrow M'$}
    \RightLabel{\scriptsize E$_2$}
    \UnaryInfC{$\boudolsubst{M}{P}{x} \boudolarrow \boudolsubst{M'}{P}{x}$}
    \endprooftree
\end{multicols}

\begin{multicols}{2}
    \prooftree
    \AxiomC{}
    \RightLabel{\scriptsize E$_3$}
    \UnaryInfC{$((\lambda x.M)P) \boudolarrow \boudolsubst{M}{P}{x}$}
    \endprooftree

    \prooftree
    \AxiomC{$(VP) \boudolarrow M \hspace{3mm} x \notin \textit{fv}(P)$}
    \RightLabel{\scriptsize E$_4$}
    \UnaryInfC{$((\boudolsubst{V}{R}{x})P) \boudolarrow \boudolsubst{M}{R}{x}$}
    \endprooftree
\end{multicols}

    \prooftree
    \AxiomC{$(\boudolsubst{M}{N}{x}) \boudolarrow M' \hspace{3mm} P \equiv (N | R) \hspace{3mm} x \notin \textit{fv}(N)$}
    \RightLabel{\scriptsize E$_5$}
    \UnaryInfC{$\boudolsubst{M}{P}{x} \boudolarrow \boudolsubst{M'}{R}{x}$}
    \endprooftree

\begin{multicols}{2}
    \prooftree
    \AxiomC{}
    \RightLabel{\scriptsize S$_1$}
    \UnaryInfC{$\boudolsubst{x}{M}{x} \boudolarrow M$}
    \endprooftree

    \prooftree
    \AxiomC{$\boudolsubst{M}{N}{x} \boudolarrow M'$}
    \RightLabel{\scriptsize S$_2$}
    \UnaryInfC{$\boudolsubst{(MP)}{N}{x} \boudolarrow (M'P)$}
    \endprooftree
\end{multicols}

    \prooftree
    \AxiomC{$\boudolsubst{M}{N}{x} \boudolarrow M' \hspace{3mm} z \neq x \hspace{3mm} z \notin \textit{fv}(N)$}
    \RightLabel{\scriptsize S$_3$}
    \UnaryInfC{$\boudolsubst{(\boudolsubst{M}{P}{z})}{N}{x} \boudolarrow \boudolsubst{M'}{P}{z}$}
    \endprooftree
    \figline
\caption{Evaluation in Boudol's calculus}
\label{boudol:evaluation}
\end{figure}

\cref{boudol:evaluation} shows the evaluation rules used by Boudol \cite{boudol1993lambda}. There is also a ``garbage collection'' rule, defined as follows.

\prooftree
\AxiomC{}
\RightLabel{$x \notin \textit{fv}(M)$}
\UnaryInfC{$\boudolsubst{M}{P}{x} \boudolarrow M$}
\endprooftree

Similarly, since $1$ denotes an empty bag, $\boudolsubst{M}{1}{x} \equiv M$.

Given a bag $P=(\finboudol{M_1}{m_1} \mid \cdots \mid \finboudol{M_k}{m_k})$, for some $k \in \mathbb{N}$, one may fetch any term inside $P$, thus we can see that evaluation is non-deterministic. In the rest of the paper, we will use $\boudolequal$ standing for the least equivalence relation containing $\boudolarrow$, and $\boudoldblarrow$ as the transitive closure of $\boudolarrow$.
\begin{example}~\label[example]{boudol:succ_ex}
    Let us consider the following example, where $I \equiv \lambda z.z$.
    \[
        \boudolsubst{(x\finboudol{x}{1})}{\finboudol{I}{2}}{x} \boudolarrow \boudolsubst{(I\finboudol{x}{1})}{\finboudol{I}{1}}{x} \boudolarrow \boudolsubst{(\boudolsubst{z}{\finboudol{x}{1}}{z})}{\finboudol{I}{1}}{x} \boudolarrow \boudolsubst{(\boudolsubst{x}{1}{z})}{\finboudol{I}{1}}{x} \boudolarrow \boudolsubst{(\boudolsubst{I}{1}{x})}{1}{z} \equiv \boudolsubst{I}{1}{z} \equiv I
    \] 
    
\end{example}

\begin{example} \label[example]{boudol:fail_ex}
    Let us consider the following example, where $I \equiv \lambda z.z$.
    \[
        \boudolsubst{(x\finboudol{x}{1})}{\finboudol{I}{1}}{x} \boudolarrow \boudolsubst{(I\finboudol{x}{1})}{1}{x} \boudolarrow \boudolsubst{(\boudolsubst{z}{\finboudol{x}{1}}{z})}{1}{x} \boudolarrow \boudolsubst{(\boudolsubst{x}{1}{z})}{1}{x} \equiv \boudolsubst{x}{1}{x}
    \]
\end{example}
Following~\cref{boudol:succ_ex} and \cref{boudol:fail_ex}, we can see a well-formed and an ill-formed term respectively. In \cref{boudol:succ_ex}, we are able to reach a weak-head normal form, whereas, in \cref{boudol:fail_ex}, we do not have any resources available to substitute $x$, meaning that no reduction is possible, therefore we enter deadlock. 

\begin{lemma}[Substitution lemma]~\label[lemma]{boudol:subst_lemma} \\
    \begin{enumerate}
        \item~\label{boudol:subst_lemma_inf}Given that $y \notin \textit{fv}(P)$, $\boudolsubst{(\boudolsubst{M}{\infboudol{Q}}{y})}{\infboudol{P}}{x} \boudolequal \boudolsubst{(\boudolsubst{M}{\infboudol{(\boudolsubst{Q}{\infboudol{P}}{x})}}{y})}{\infboudol{P}}{x}$
        \item~\label{boudol:subst_lemma_fin}Given that $y \notin \textit{fv}(P)$, $\boudolsubst{(\boudolsubst{M}{\finboudol{Q}{n}}{y})}{\finboudol{P}{m}}{x} \boudolequal \boudolsubst{(\boudolsubst{M}{\finboudol{(\boudolsubst{Q}{\finboudol{P}{m-s}}{x})}{n}}{y})}{\finboudol{P}{s}}{x}$
    \end{enumerate}
\end{lemma}
\begin{proof}
    By structural induction on $M$.
\end{proof}

\begin{lemma}~\label[lemma]{boudol:subst_app_lemma}  \\
    \begin{enumerate}
        \item~\label{boudol:subst_app_lemma_inf} If $\boudolsubst{(M\infboudol{Q})}{\infboudol{P}}{x} \boudolarrow M_1$ and $(\boudolsubst{M}{\infboudol{P}}{x} \infboudol{(\boudolsubst{Q}{\infboudol{P}}{x})}) \boudolarrow M_2$ then $M_1 \boudolequal M_2$. 
        \item~\label{boudol:subst_app_lemma_fin} If $\boudolsubst{(M\finboudol{Q}{n})}{\finboudol{P}{m}}{x} \boudolarrow M_1$ and $(\boudolsubst{M}{\finboudol{P}{s}}{x} \finboudol{(\boudolsubst{Q}{\finboudol{P}{m-s}}{x})}{n}) \boudolarrow M_2$ then $M_1 \boudolequal M_2$.
    \end{enumerate}
\end{lemma}
\begin{proof} 
    By structural induction on $M$.
\end{proof}

\section{Term Expansion}
Our main goal is to relate $\lambda$-terms in a calculus with explicit substitutions with terms of Boudol's calculus. To achieve this, we will use the notion of \textit{term expansion} \cite{florido2004linearization}. 

Expansion \cite{florido2004linearization} was firstly introduced to study the relation between terms typable using intersection
types~\cite{coppo1980extension} and terms typable using simple types, through the process of replacing each occurrence of a variable in a $\lambda$-term by a new variable, thus sustaining the process of linearisation of strongly normalising $\lambda$-terms. It was inspired by the seminal work on linearisation from Kfoury \cite{kfoury2000linearization},
where a new calculus was defined in which, instead of requiring terms to be syntactically linear, a weaker linearity condition was imposed. Kfoury
defined the $\lambda^{\wedge}$-calculus, where applications were of the form $M.P_1 \wedge \ldots \wedge P_n$, and a new notion of reduction, which, when the redex was
of the form $(\lambda x.M).P_1 \wedge \ldots \wedge P_n$, with $x$ occuring $n$ times free is $M$, replaced the $ith$ occurrence of $x$ by $P_i$. 
Kfoury then defined a notion of contraction of terms and reductions of this calculus into the $\lambda$-calculus. The set of well-formed terms of the new calculus (corresponding to a linear version of a $\lambda$-term), were the ones that could be contracted into the $\lambda$-calculus. Term expansion simplified this previous framework by defining linearisation inside the $\lambda$-calculus from non-linear terms into linear ones and using standard $\beta$-reduction. 

More recent applications of {\em term expansion} include the relation with substructural type systems such as relevant and ordered type systems~\cite{alves2022structural}, and quantitative types~\cite{alves2022quantitative}. An important note is that  both the original notion of term expansion and the new one defined in this paper, define a relation and not a function, thus expansion should not be confused with a term transformation algorithm.

Following this previous work, the new term expansion defined here relates $\lambda$-terms with explicit substitutions with terms in Boudol's calculus, while inferring the amount of times each resource is available. Although Boudol uses associative, commutative and non-idempotent intersection types with the universal type $\omega$, to accomplish our goal we divide expansion into two definitions: one using associative, commutative and idempotent intersection types, and another one using associative, commutative and non-idempotent intersection types, to deal with infinite and finite multiplicities respectively, eliminating the universal type -- this will be explained in more detail further on.

Therefore, we will use associative, commutative and idempotent intersection types to relate $\lambda$-terms with terms in $\infboudol{\Lambda}$, and associative, commutative and non-idempotent intersection types to relate $\lambda$-terms with terms in $\Lambda$. Before that, we will provide a brief introduction to intersection types.

\subsection{Intersection types}

Intersection types originate in the works of Barendregt, Coppo and Dezani \cite{barendregt1983filter, coppo1980extension}. Intersection type systems without the universal type $\omega$, as presented in \cite{coppo1980extension}, give a characterization of the strongly normalisable terms, in the sense that a term is typed in an intersection type system without $\omega$ if and only if it is strongly normalisable.

In an intersection type system, variables can be assigned different types, unlike the Curry Simple Type System in which each variable is assigned a single type. For example,  the term $\lambda x. xx$, is not typable in the Curry Type System, but has type $(\alpha \cap \alpha \to \beta) \to \beta$ in an Intersection Type System.

\begin{definition} \label[definition]{exp_inter:types}
    Let $\alpha$ range over an infinite set of type variables:
    \[ (\cap\text{-types}) \ \sigma ::= \alpha \mid \sigma_1 \cap \cdots \cap \sigma_n \to \sigma \]
\end{definition}
The original Coppo-Dezani Intersection Type System \cite{coppo1980extension} considers the intersection operator to be associative, commutative and idempotent, although there are other works which consider non-idempotent intersections \cite{bucciarelli2017non,florido2004linearization,kfoury2000linearization}. We will henceforth use ACI-intersection to denote associative, commutative and idempotent intersections, i.e. $\sigma \cap \sigma = \sigma$, and AC-intersection to denote associative, commutative and non-idempotent intersections, $\sigma \cap \sigma \neq \sigma$.

In this paper, the symbols $\alpha$ and $\beta$ will denote type-variables and $\tau$, $\sigma$ and $\psi$ will denote types. The type constructor $\rightarrow$ is assumed to be right associative, and  we assume that $\cap$ binds stronger than $\rightarrow$. All symbols can appear indexed.

Some systems with intersection types consider a universal type constant  $\omega$, corresponding to the empty intersection \cite{barendregt1983filter, coppo1980extension}. To briefly explain the role of this universal type $\omega$, it is essentially used to assign a type to every term, including non-normalising terms. For example, consider the term $(\lambda xy.y)\Omega$, where $\Omega \equiv (\lambda x. xx)(\lambda x. xx)$. $\Omega$ does not have a normal form, since it reduces infinitely to itself ($\Omega \to \Omega \to \cdots$). Nevertheless, although $(\lambda xy.y)\Omega$ contains $\Omega$ as a subterm, it does have a normal form, $(\lambda xy.y)\Omega \to \lambda y. y$. Hence, $(\lambda xy.y)\Omega$ has type $\alpha \to \alpha$, where $\lambda xy.y$ has type $\omega \to \alpha \to \alpha$ and $\Omega$ has type $\omega$. 

\subsection{Expansion and infinite multiplicities}
We now extend the notion of expansion  to relate our $\lambda$-calculus with explicit substitutions with Boudol's resourse calculus with infinite multiplicities. Let us start by defining the notion of expansion context.
\begin{definition} \label[definition]{exp:inf_exp_context}
    An \textit{expansion context A} is any finite set of variable expansions of the form \\
    $A = \{x_1 : \tau_1, \dots, x_n: \tau_n \}$
    where the variables $ x_1, \dots, x_n$ are pairwise distinct and $\tau_1, \dots, \tau_n$ are types.
\end{definition}

Given two expansion contexts, we define an operation that appends them, in the following way. 

\begin{definition} \label[definition]{exp:exp_app_inf_context}
    Let $A_1$ and $A_2$ be two expansion contexts. Then $A_1 \:\&\: A_2$ is a new context such that $x: \tau \in A_1 \:\&\: A_2$ if and only if
    \[
    \tau =
        \begin{cases}
        \tau_1 \cap \tau_2 & \text{if } x : \tau_1 \in A_1 \text{ and } x : \tau_2 \in A_2 \\
        \tau_1         & \text{if } x : \tau_1 \in A_1 \text{ and } \neg\exists \tau . x : \tau \in A_2 \\
        \tau_2         & \text{if } x : \tau_2 \in A_2 \text{ and } \neg\exists \tau . x : \tau \in A_1
        \end{cases}
    \]
\end{definition}
Whenever we write $A \:\&\: \{x : \tau\}$, we assume that $x$ does not occur in $A$.

Let us now formalise the notion of term expansion for ACI-intersection types.

\begin{definition}[Expansion in $\infboudol{\Lambda}$] \label[definition]{exp:def_exp_inf}
    Given a pair $M:\sigma$, where $M$ is a $\lambda$xx-term and $\sigma$ is an intersection type, and a term $N \in \infboudol{\Lambda}$ and an expansion context $A$, we define the  \textit{expansion} relation \\
    $\infexpansion{M:\sigma}{(N,A)}$, as follows.
\[
\begin{array}{rcll}
    \infexpansion{x:\tau} & & {(x, \{x : \tau\})} \\
    \infexpansion{\lambda x. M: \tau_1 \cap \dots \cap \tau_n \to \sigma} & & {(\lambda x. \transboudol{M}, A)} \\
    & &\text{if } x \in \textit{fv}(M) \text{ and} \\
    & &\infexpansion{M: \sigma}{(\transboudol{M}, A \:\&\: \{x : \tau_1 \cap \dots \cap \tau_n \})} \\
    \infexpansion{MN:\sigma}&&{(\transboudol{M} (\infboudol{P_1} \mid \dots \mid \infboudol{P_k}), A_0 \:\&\: A_1 \:\&\: \cdots \:\&\: A_k)} \\
    & &\text{if for some } k>0 \text{ and } \tau_1, \dots, \tau_k \text{ such that} \\
    & &\infexpansion{M:\tau_1 \cap \dots \cap \tau_k \to \sigma}{(\transboudol{M}, A_0)} \text{ and} \\ 
    & &\infexpansion{N: \tau_i}{(P_i, A_i)} \text{ for } 1 \leq i \leq k \\
    \infexpansion{\expsubst{M}{x}{N}:\sigma}&&{(\boudolsubst{\transboudol{M}}{(\infboudol{P_1} \mid \dots \mid \infboudol{P_k}}{x}, A_0 \:\&\: A_1 \:\&\: \cdots \:\&\: A_k)} \\
    & &\text{if for some } k>0, \infexpansion{M: \sigma}{(\transboudol{M}, A_0}), \\
    & &\text{where } A_0(x) = \tau_1 \cap \cdots \cap \tau_k, \text{ and} \\ 
    & &\infexpansion{N: \tau_i}{(P_i, A_i)} \text{ for } 1 \leq i \leq k 

    \end{array}
\]
\end{definition}
We will write $\infexpansion{M:\sigma}{N}$, if $A = \emptyset$.

Note that in the previous definition of expansion (as it happens with the original definition in \cite{florido2004linearization}), when using a pair $(M : \sigma)$, where $M$ is a term and $\sigma$ is an intersection type, one does not require $\sigma$ to be a type derivable for $M$ in an intersection type system. In fact, expansion itself already implicitly supplies a type checking job in the way types are used in its inductive definition. 

\begin{example} \label[example]{exp:exp_inf_example}
    Let $I \equiv \lambda z.z$ and $\sigma \equiv \alpha \to \alpha$. 
    
    We will show how to calculate the expansion of $((\lambda x. xx)I: \alpha \to \alpha)$.
    
    Firstly, we have $\infexpansion{x:\sigma \to \sigma}{(x, \{ x : \sigma \to \sigma \})}$ and $\infexpansion{x:\sigma}{(x, \{ x : \sigma \})}$. Thus, \\
    $\infexpansion{xx:\sigma}{((x \infboudol{x}), \{ x : (\sigma \to \sigma) \cap \sigma \})}$, concluding $\infexpansion{\lambda x.xx: ((\sigma \to \sigma) \cap \sigma) \to \sigma}{\lambda x. (x \infboudol{x})}$.
    
    Now, it is easy to show that $\infexpansion{I: \sigma \to \sigma}{I}$ and $\infexpansion{I: \sigma}{I}$.
    
    Therefore, this results in $\infexpansion{(\lambda x. xx)I: \sigma}{((\lambda x. (x \infboudol{x}))\infboudol{I})}$.
\end{example}

We will now show that expansion using ACI-intersection types preserves weak-head reduction, as the following diagram indicates.
\[
\begin{tikzcd}[row sep=3em, column sep=4em]
M \arrow[r,swap,"\text{bxgc}"] \arrow[d,"\mathcal{E}^\infty"'] &
M' \arrow[d,"\mathcal{E}^\infty"] \\[1em]
M^*
  \arrow[r,
    "{\mathclap{\;\;=\!=\!=\!\boudolequal\;}}",
    phantom
  ]
&
M''
\end{tikzcd}
\]
\begin{theorem}[Expansion and Infinite Multiplicities] \label[theorem]{exp:theo_exp_inf}
    Given a $\lambda$x-term $M$ and an ACI-intersection type $\sigma$, such that $\infexpansion{M:\sigma}{(\transboudol{M}, A_1)}$, if $M \singlearrow M'$ then $\transboudol{M} \boudolequal M'' \text{ and} \infexpansion{M' : \sigma}{(M'', A_2)}$, where $A_2 \subseteq A_1$.
\end{theorem}
\begin{proof}
    By structural induction on the reduction $\singlearrow$.
    \begin{itemize}
    \item Base case:
    \begin{itemize}
        \item Given $\infexpansion{(\lambda x.M)N: \sigma}{(((\lambda x. \transboudol{M})\infboudol{P}), A_0 \& A_1 \& \cdots \& A_k)}$, because
            \[ \infexpansion{\lambda x. M : \tau_1 \cap \cdots \cap \tau_k \to \sigma}{(\lambda x. \transboudol{M}, A_0)} \]
            where $\infexpansion{M:\sigma}{(\transboudol{M}, A_0 \& \{x : \tau_1 \cap \cdots \cap \tau_k\})}$, for some $k > 0$, and $\infexpansion{N: \tau_i}{(P_i, A_i)}$ with $1 \leq i \leq k$.

            We know that $(\lambda x. M)N \singlearrow \expsubst{M}{x}{N}$, and \\
            $\infexpansion{\expsubst{M}{x}{N}: \sigma}{(\boudolsubst{\transboudol{M}}{\infboudol{P}}{x}, A_0 \& A_1 \& \cdots \& A_k)}$.

            We also know that $((\lambda x. \transboudol{M})\infboudol{P}) \boudolarrow \boudolsubst{\transboudol{M}}{\infboudol{P}}{x}$ by rule E$_3$.

        \item Given $\infexpansion{\expsubst{x}{x}{N}:\sigma}{(\boudolsubst{x}{\infboudol{P}}{x}, A)}$, because $\infexpansion{x : \sigma}{x} \text{ and } \infexpansion{N : \sigma}{(P, A)}$.

            We have $\expsubst{x}{x}{N} \singlearrow N$, and $\infexpansion{N : \sigma}{(P, A)}$.

            We also know that $\boudolsubst{x}{\infboudol{P}}{x} \boudolarrow \boudolsubst{P}{\infboudol{P}}{x}$ by rule S$_1$.

            Since $x \notin \textit{fv}(P)$, by the garbage collection rule, we have $\boudolsubst{P}{\infboudol{P}}{x} \boudolarrow P$.
        
        \item Given $\infexpansion{\expsubst{M}{x}{N}: \sigma}{(\boudolsubst{\transboudol{M}}{\infboudol{P}}{x}, A_0 \& A_1 \& \cdots \& A_k)}$, because $\infexpansion{M:\sigma}{(\transboudol{M}, A_0)}$ and $A(x) = \tau_1 \cap \cdots \tau_k$, $\infboudol{N:\tau_i}{(P_i, A_i)}$, where $1 \leq i \leq k$.

            We have $\expsubst{M}{x}{N} \singlearrow M$ if $x \notin \textit{fv}(M)$, and $\infexpansion{M:\sigma}{(\transboudol{M}, A_0)}$.

            Since $x \notin \textit{fv}(\transboudol{M})$ then, by garbage collection, $\boudolsubst{\transboudol{M}}{\infboudol{P}}{x} \boudolarrow \transboudol{M}$, and $A_0 \subseteq A_0 \& \cdots \& A_k$.

        \item Given $\infexpansion{\expsubst{(M_1 M_2)}{x}{N}: \sigma}{(\boudolsubst{(\transboudol{M_1}\infboudol{Q})}{\infboudol{P}}{x}, A \& B_1 \& \cdots \& B_{k_2})}$, where $x \in \textit{fv}(M_1)$ and $x \in \textit{fv}(M_2)$, because $\infexpansion{M_1M_2: \sigma}{((\transboudol{M_1}\infboudol{Q}), A_0 \& A_1 \& \cdots \& A_{k_1})}$ where, for some $k_1>0$, $\infexpansion{M_1: \tau_1 \cap \cdots \cap \tau_{k_1} \to \sigma}{(\transboudol{M_1}, A_0)}  \text{ and } \infexpansion{M_2 : \tau_i}{(Q_i, A_i)} \text{, where } 1 \leq i \leq k_1$
            and, for some $k_2 > 0$ and $x : \sigma_1 \cap \cdots \cap \sigma_{k_2}$ with $A = A_1 \& \cdots \& A_{k_1}$.
            
            Also, $\infexpansion{N : \sigma_j}{(P_j, B_j)} \text{, where } 1 \leq j \leq k_2$.

            We know that $\expsubst{(M_1 M_2)}{x}{N} \singlearrow \expsubst{M_1}{x}{N} \expsubst{M_2}{x}{N}$.

            Now, $\infexpansion{\expsubst{M_1}{x}{N} \expsubst{M_2}{x}{N} : \sigma}{((\boudolsubst{\transboudol{M_1}}{\infboudol{P}}{x}\infboudol{(\boudolsubst{Q}{\infboudol{P}}{x})}), A \& B_1 \& \cdots \& B_{k_2})}$.
            
            By \cref{boudol:subst_app_lemma}.\ref{boudol:subst_app_lemma_inf}, we have that $\boudolsubst{(\transboudol{M_1}\infboudol{Q})}{\infboudol{P}}{x} \boudolequal (\boudolsubst{\transboudol{M_1}}{\infboudol{P}}{x}\infboudol{(\boudolsubst{Q}{\infboudol{P}}{x})})$. 

        \item Given $\infexpansion{\expsubst{(M_1 M_2)}{x}{N}: \sigma}{(\boudolsubst{(\transboudol{M_1}\infboudol{Q})}{\infboudol{P}}{x}, A \& B_1 \& \cdots \& B_{k_2})}$, where $x \in \textit{fv}(M_1)$ and $x \notin \textit{fv}(M_2)$, because $\infexpansion{M_1M_2: \sigma}{((\transboudol{M_1}\infboudol{Q}), A_0 \& A_1 \& \cdots \& A_{k_1})}$ where, for some $k_1>0$, $\infexpansion{M_1: \tau_1 \cap \cdots \cap \tau_{k_1} \to \sigma}{(\transboudol{M_1}, A_0)}  \text{ and } \infexpansion{M_2 : \tau_i}{(Q_i, A_i)} \text{, where } 1 \leq i \leq k_1$
            and, for some $k_2 > 0$ and $x : \sigma_1 \cap \cdots \cap \sigma_{k_2}$, with $A = A_0 \& \cdots \& A_{k_1}$.
            
            Also $\infexpansion{N : \sigma_j}{(P_j, B_j)} \text{, where } 1 \leq j \leq k_2$.

            Since we know that $x \notin \textit{fv}(M_2)$, by rule S$_2$, we have $\boudolsubst{(\transboudol{M_1}\infboudol{Q})}{\infboudol{P}}{x} \boudoldblarrow (M_1{'}\infboudol{Q})$ because $\boudolsubst{\transboudol{M_1}}{\infboudol{P}}{x} \boudoldblarrow M_1{'}$.
    
            We know that $\expsubst{(M_1 M_2)}{x}{N} \singlearrow \expsubst{M_1}{x}{N} M_2$. And we have \\
            $\infexpansion{\expsubst{M_1}{x}{N}M_2:\sigma}{((\boudolsubst{\transboudol{M_1}}{\infboudol{P}}{x} \infboudol{Q}), A \& B_1 \& \cdots \& B_{k_2})}$. Since we know that \\
            $\boudolsubst{\transboudol{M_1}}{\infboudol{P}}{x} \boudoldblarrow M_1{'}$ then, by rule E$_1$, $(\boudolsubst{\transboudol{M_1}}{\infboudol{P}}{x} \infboudol{Q}) \boudoldblarrow (M_1{'}\infboudol{Q})$.

        \item Given $\infexpansion{\expsubst{(M_1 M_2)}{x}{N}: \sigma}{(\boudolsubst{(\transboudol{M_1}\infboudol{Q})}{\infboudol{P}}{x}, A \& B_1 \& \cdots \& B_{k_2})}$, where $x \notin \textit{fv}(M_1)$ and $x \in \textit{fv}(M_2)$, because $\infexpansion{M_1M_2: \sigma}{((\transboudol{M_1}\infboudol{Q}), A_0 \& A_1 \& \cdots \& A_{k_1})}$ where, for some $k_1>0$, $\infexpansion{M_1: \tau_1 \cap \cdots \cap \tau_{k_1} \to \sigma}{(\transboudol{M_1}, A_0)}  \text{ and } \infexpansion{M_2 : \tau_i}{(Q_i, A_i)} \text{, where } 1 \leq i \leq k_1$
            and, for some $k_2 > 0$ and $x : \sigma_1 \cap \cdots \cap \sigma_{k_2}$, with $A = A_0 \& \cdots \& A_{k_1}$.
            
            Also $\infexpansion{N : \sigma_j}{(P_j, B_j)} \text{, where } 1 \leq j \leq k_2$.

            Since $x \notin \textit{fv}(M_1)$, then $\transboudol{M_1} \equiv \lambda y. \transboudol{N_1}$, and we have \\
            $\boudolsubst{((\lambda y. \transboudol{N_1})\infboudol{Q})}{\infboudol{P}}{x} \boudolarrow \boudolsubst{(\boudolsubst{\transboudol{N_1}}{\infboudol{Q}}{y})}{\infboudol{P}}{x}$ by rule E$_2$.

            We know that $\expsubst{(M_1 M_2)}{x}{N} \singlearrow M_1 \expsubst{M_2}{x}{N}$. And \\
            $\infexpansion{M_1 \expsubst{M_2}{x}{N}:\sigma}{((\transboudol{M_1}\infboudol{(\boudolsubst{Q}{\infboudol{P}}{x})}), A \& B_1 \& \cdots \& B_{k_2})}$. From which we get, since $\transboudol{M_1} \equiv \lambda y. \transboudol{N_1}$,
            $((\lambda y. \transboudol{N_1})\infboudol{(\boudolsubst{Q}{\infboudol{P}}{x})}) \boudolarrow \boudolsubst{\transboudol{N_1}}{\infboudol{(\boudolsubst{Q}{\infboudol{P}}{x})}}{y}$. 

            Since $x \notin \textit{fv}(M_1)$, then $\boudolsubst{(\boudolsubst{\transboudol{N_1}}{\infboudol{Q}}{y})}{\infboudol{P}}{x} \boudolequal \boudolsubst{\transboudol{N_1}}{\infboudol{(\boudolsubst{Q}{\infboudol{P}}{x})}}{y}$.
    \end{itemize}
    \item Inductive case: $\infexpansion{MN:\sigma}{((\transboudol{M}\infboudol{Q}), A_0 \& A_1 \& \cdots \& A_k)}$ and \\
    $\infexpansion{(\expsubst{M}{x}{N_1})N_2 : \sigma}{((\boudolsubst{\transboudol{M}}{\infboudol{P}}{x}\infboudol{Q}), A_0 \& A_1 \& \cdots \& A_{k_1} \& B_1 \& \cdots \& B_{k_2})}$ follow easily by induction.
\end{itemize}

\end{proof}

\begin{example} \label[example]{exp:theo_exp_inf_example}
    Starting from the term $\expsubst{(xx)}{x}{I}$, where $I \equiv \lambda z.z$ and $\sigma \equiv \alpha \to \alpha$.

    We know that $\expsubst{(xx)}{x}{I} \singlearrow \expsubst{x}{x}{I} (\expsubst{x}{x}{I})$ and, from $\infexpansion{x: \sigma \to \sigma}{(x, \{x:\sigma \to \sigma\})}$ and $\infexpansion{x : \sigma}{(x, \{x:\sigma\})}$, we have $\infexpansion{xx:\sigma}{((x\infboudol{x}), \{x : (\sigma \to \sigma) \cap \sigma)}$, and from \\
    $\infexpansion{I:\sigma \to \sigma}{I}$ and $\infexpansion{I:\sigma}{I}$, thus $\infexpansion{\expsubst{(xx)}{x}{I}:\sigma}{\boudolsubst{(x\infboudol{x})}{\infboudol{I}}{x}}$

    \[
    \begin{array}{rcll}
        \boudolsubst{(x\infboudol{x})}{\infboudol{I}}{x} \boudolarrow \boudolsubst{(I\infboudol{x})}{\infboudol{I}}{x} \boudolarrow \boudolsubst{(\boudolsubst{z}{\infboudol{x}}{z})}{\infboudol{I}}{x} &\boudolarrow& \boudolsubst{(\boudolsubst{x}{\infboudol{x}}{z})}{\infboudol{I}}{x} \\
        &\boudolarrow& \boudolsubst{(\boudolsubst{I}{\infboudol{I}}{x})}{\infboudol{x}}{z} \boudoldblarrow I 
    \end{array}
    \]

    We also have that, from $\infexpansion{x:\sigma \to \sigma}{(x, \{x : \sigma \to \sigma \})}$ and $\infexpansion{I:\sigma \to \sigma}{I}$, we have $\infexpansion{\expsubst{x}{x}{I}:\sigma \to \sigma}{\boudolsubst{x}{\infboudol{I}}{x}}$, and, from $\infexpansion{x : \sigma}{(x, \{x : \sigma\})}$ and $\infexpansion{I : \sigma}{I}$, we have $\infexpansion{\expsubst{x}{x}{I}:\sigma}{\boudolsubst{x}{\infboudol{I}}{x}}$. Thus $\infexpansion{\expsubst{x}{x}{I} \expsubst{x}{x}{I}: \sigma}{(\boudolsubst{x}{\infboudol{I}}{x}\infboudol{(\boudolsubst{x}{\infboudol{I}}{x})})}$.

    \[
    \begin{array}{rcll}
        (\boudolsubst{x}{\infboudol{I}}{x}\infboudol{(\boudolsubst{x}{\infboudol{I}}{x})}) &\boudolarrow& (\boudolsubst{I}{\infboudol{I}}{x}\infboudol{(\boudolsubst{x}{\infboudol{I}}{x})}) \\
        &\boudolarrow& \boudolsubst{(\boudolsubst{z}{\infboudol{(\boudolsubst{x}{\infboudol{I}}{x})}}{z})}{\infboudol{I}}{x} \\
        &\boudolarrow& \boudolsubst{(\boudolsubst{(\boudolsubst{x}{\infboudol{I}}{x})}{\infboudol{(\boudolsubst{x}{\infboudol{I}}{x})}}{z})}{\infboudol{I}}{x} \\
        &\boudolarrow& \boudolsubst{(\boudolsubst{(\boudolsubst{I}{\infboudol{I}}{x})}{\infboudol{(\boudolsubst{x}{\infboudol{I}}{x})}}{z})}{\infboudol{I}}{x} \boudoldblarrow I
    \end{array}
    \]
    Therefore $\boudolsubst{(x\infboudol{x})}{\infboudol{I}}{x} \boudolequal (\boudolsubst{x}{\infboudol{I}}{x}\infboudol{(\boudolsubst{x}{\infboudol{I}}{x})})$.
\end{example}

\begin{corollary} \label[corollary]{exp:cor_exp_inf}
    Given a $\lambda$x-term $M$ and an ACI-intersection type $\sigma$, such that $\infexpansion{M: \sigma}{\transboudol{M}}$, if $M \dblarrow V$ then $\transboudol{M} \boudolequal V'$ and $\infexpansion{V:\sigma}{V'}$, where $V$ and $V'$ are weak-head normal forms.
\end{corollary}
\begin{example} \label[example]{exp:cor_exp_inf_example}
    Following the term in \cref{exp:theo_exp_inf_example}, we have 
    \[
    \begin{array}{rcll}
         \expsubst{(xx)}{x}{I} &\singlearrow& \expsubst{x}{x}{I} (\expsubst{x}{x}{I}) \\
         &\singlearrow& I(\expsubst{x}{x}{I}) \singlearrow \expsubst{y}{y}{(\expsubst{x}{x}{I})} \singlearrow \expsubst{x}{x}{I} \singlearrow I 
    \end{array}
    \]
    And we already know that $\boudolsubst{(x\infboudol{x})}{\infboudol{I}}{x} \boudoldblarrow I$ and $\infexpansion{I:\sigma}{I}$.
\end{example}

\subsection{Expansion and finite multiplicities}
Let us start by redefining the notion of expansion context, to consider finite multiplicities.
\begin{definition} \label[definition]{exp:fin_exp_context}
    An \textit{expansion context A} is any finite set of variable expansions \\
    $A = \{x_1 : \tau_1{^{m_1}}, \dots, x_n: \tau_n{^{m_n}} \}$
    where the variables $x_1, \dots, x_n $ are pairwise distinct, $\tau_1, \dots, \tau_n$ are types and $m_1, \dots,  m_n$ are multiplicities.
\end{definition}

We now define an operation that appends two expansion contexts. 

\begin{definition} \label[definition]{exp:exp_app_fin_context}
    Let $A_1$ and $A_2$ be two expansion contexts. Then $A_1 \otimes A_2$ is a new context such that $x: \tau \in A_1 \otimes A_2$ if and only if
    \[
    \tau =
        \begin{cases}
        (\tau_1 \cap \tau_2){^{m_1 + m_2}} & \text{if } x : \tau_1{^{m_1}} \in A_1 \text{ and } x : \tau_2{^{m_2}} \in A_2 \\
        \tau_1{^{m_1}}         & \text{if } x : \tau_1{^{m_1}} \in A_1 \text{ and } \neg\exists \tau . x : \tau \in A_2 \\
        \tau_2{^{m_2}}         & \text{if } x : \tau_2{^{m_2}} \in A_2 \text{ and } \neg\exists \tau . x : \tau \in A_1
        \end{cases}
    \]
\end{definition}

Whenever we write $A \otimes \{x : \tau\}$, we assume that $x$ does not occur in $A$. We are now able to formalise the notion of term expansion for AC-intersection types.

\begin{definition}[Expansion in $\Lambda$] \label[label]{exp:def_exp_fin}
    Given a pair $M:\sigma$, where $M$ is a $\lambda$-term and $\sigma$ an intersection type, and a term $N \in \Lambda$ and an expansion context $A$, we define the \textit{expansion} relation $\finexpansion{M:\sigma}{(N, A)}$, as follows.
\[
\begin{array}{rcll}
    \finexpansion{x:\tau^1}&&{(x, \{x : \tau^1\})} \\
    \finexpansion{\lambda x. M: \tau_1{^{m_1}} \cap \dots \cap \tau_n{^{m_n}} \to \sigma}&&{(\lambda x. \transboudol{M}, A)} \\
    & &\text{if } x \in \textit{fv}(M) \text{ and } \\
    & &\finexpansion{M: \sigma}{(\transboudol{M}, A \otimes \{x : \tau_1{^{m_1}} \cap \dots \cap \tau_n{^{m_n}} \})} \\
    \finexpansion{MN:\sigma}&&{((\transboudol{M} (\finboudol{P_1}{m_1} \mid \dots \mid \finboudol{P_k}{m_k})), A_0 \otimes A_1 \otimes \dots \otimes A_k)} \\
    & &\text{if for some } k>0 \text{ and } \tau_1{^{m_1}}, \dots, \tau_k{^{m_k}} \text{ such that } \\
    & &\finexpansion{M:\tau_1{^{m_1}} \cap \dots \cap \tau_k{^{m_k}} \to \sigma}{(\transboudol{M}, A_0)} \text{ and} \\ 
    & & \finexpansion{N: \tau_i{^1}}{(P_i, A_i)} \text{ for } 1 \leq i \leq k \\
    \finexpansion{\expsubst{M}{x}{N}:\sigma}&&{(\boudolsubst{\transboudol{M}}{(\finboudol{P_1}{m_1} \mid \dots \mid \finboudol{P_k}{m_k})}{x}, A_0 \otimes A_1 \otimes \dots \otimes A_k)} \\
    & &\text{if for some } k>0, \ \finexpansion{M: \sigma}{(\transboudol{M}, A_0)}, \\
    & & \text{where } A_0(x) = \typemult{(\tau_1 \cap \cdots \cap \tau_k)}{m_1 + \cdots + m_k} \text{, and } \\
    & & \finexpansion{N: \tau_i{^1}}{(P_i, A_i)} \text{ for } 1 \leq i \leq k 
    \end{array}
\]
\end{definition}

Note that we will write $\finexpansion{M:\sigma}{N}$ if $A = \emptyset$.
\begin{example} \label[example]{exp:exp_fin_example}
    Let $I \equiv \lambda z.z$ and $\sigma \equiv \alpha \to \alpha$. 
    
    We will show how to calculate the expansion of $((\lambda x. xx)I: \alpha \to \alpha)$.
    
    Firstly, we have $\finexpansion{x:\sigma \to \sigma}{(x, \{x : \typemult{(\sigma \to \sigma)}{1}\})}$ and $\finexpansion{x:\sigma}{(x, \{x: \typemult{\sigma}{1}\})}$. Thus, $\finexpansion{xx:\sigma}{((x \finboudol{x}{1}), \{x : (\sigma \to \sigma) \cap \sigma\})}$, concluding $\finexpansion{\lambda x.xx: ((\sigma \to \sigma) \cap \sigma) \to \sigma}{\lambda x. (x \finboudol{x}{1})}$.
    
    Now, it is easy to show that $\finexpansion{I: \sigma \to \sigma}{I}$, where $\finexpansion{y:\sigma \to \sigma}{(y, \{y : \typemult{(\sigma \to \sigma)}{1}\})}$, and $\finexpansion{I: \sigma}{I}$, where $\finexpansion{y : \sigma}{(y, \{y : \typemult{\sigma}{1}\})}$. Therefore, this results in $\finexpansion{(\lambda x. xx)I: \sigma}{((\lambda x. (x \finboudol{x}{1}))\finboudol{I}{2})}$.
\end{example}

Now, we will now show that expansion using AC-intersection types preserves weak-head reduction, as the following diagram indicates.

\[
\begin{tikzcd}[row sep=3em, column sep=4em]
M \arrow[r,swap,"\text{bxgc}"] \arrow[d,"\mathcal{E}"'] &
M' \arrow[d,"\mathcal{E}"] \\[1em]
M^*
  \arrow[r,
    "{\mathclap{\;\;=\!=\!=\!\boudolequal\;}}",
    phantom
  ]
&
M''
\end{tikzcd}
\]
\begin{theorem}[Expansion and Finite Multiplicities] \label[theorem]{exp:theo_exp_fin}
    Given a $\lambda$x-term $M$ and an AC-intersection type $\sigma$, such that $\finexpansion{M:\sigma}{(\transboudol{M}, A_1)}$, if $M \singlearrow M'$ then $\transboudol{M} \boudolequal M''$ and $\finexpansion{M' : \sigma}{(M'', A_2)}$, where $A_2 \subseteq A_1$. 
\end{theorem}
\begin{proof}
    By induction on the definition on the reduction $\singlearrow$.
    \begin{itemize}
    \item Base case:
    \begin{itemize}
        \item Given $\finexpansion{(\lambda x.M)N: \sigma}{(((\lambda x. \transboudol{M})\finboudol{P}{m_1 + \cdots + m_k}), A_0 \otimes A_1 \otimes \cdots \otimes A_k)}$, because
            \[ \finexpansion{\lambda x. M : \typemult{\tau_1}{m_1} \cap \cdots \cap \typemult{\tau_k}{m_k} \to \sigma}{(\lambda x. \transboudol{M}, A_0)} \]
            where $\finexpansion{M:\sigma}{(\transboudol{M}, A_0 \otimes \{x : \typemult{\tau_1}{m_1} \cap \cdots \cap \typemult{\tau_k}{m_k} \})}$, for some $k > 0$ \\
            and $\finexpansion{N: \tau_i}{(P_i, A_i)}$ with $1 \leq i \leq k$.

            We know that $(\lambda x. M)N \singlearrow \expsubst{M}{x}{N}$, and \\
            $\finexpansion{\expsubst{M}{x}{N}: \sigma}{(\boudolsubst{\transboudol{M}}{\finboudol{P}{m_1 + \cdots + m_k}}{x}, A_0 \otimes A_1 \otimes \cdots \otimes A_k)}$.

            We also know that $((\lambda x. \transboudol{M})\finboudol{P}{m_1 + \cdots + m_k}) \boudolarrow \boudolsubst{\transboudol{M}}{\finboudol{P}{m_1 + \cdots + m_k}}{x}$ by rule E$_3$.

        \item Given $\finexpansion{\expsubst{x}{x}{N}:\sigma}{(\boudolsubst{x}{P}{x}, A)}$, because $\finexpansion{x : \sigma}{x}$ and $\finexpansion{N : \sigma}{(P, A)}$.

            We have $\expsubst{x}{x}{N} \singlearrow N$, and $\finexpansion{N : \sigma}{(P,A)}$.

            We also know that $\boudolsubst{x}{P}{x} \boudolarrow P$ by rule S$_1$.

        \item Given $\finexpansion{\expsubst{(M_1 M_2)}{x}{N}: \sigma}{(\boudolsubst{(\transboudol{M_1}\finboudol{Q}{m_1 + \cdots + m_{k_1}})}{\finboudol{P}{n_1 + \cdots + n_{k_2}}}{x}, A \otimes B_1 \otimes \cdots \otimes B_{k_2})}$, where $x \in \textit{fv}(M_1)$ and $x \in \textit{fv}(M_2)$, because $\finexpansion{M_1M_2: \sigma}{(\transboudol{M_1}\finboudol{Q}{m_1 + \cdots + m_{k_1}}, A_0 \otimes A_1 \otimes \cdots \otimes A_{k_1})}$
            where, for some $k_1>0$, $\finexpansion{M_1: \typemult{\tau_1}{m_1} \cap \cdots \cap \typemult{\tau_{k_1}}{m_{k_1}} \to \sigma}{(\transboudol{M_1}, A_0)}$ and $\finexpansion{M_2 : \tau_i}{(Q_i, A_i)}$, where $1 \leq i \leq k_1$
            and, for some $k_2 > 0$, $A(x) = \typemult{(\sigma_1 \cap \cdots \cap \sigma_{k_2})}{n_1 + \cdots + n_{k_2}}$, with $A = A_0 \otimes \cdots \otimes A_{k_1}$. And \\
            $\finexpansion{N : \sigma_j}{(P_j, B_j)} \text{, where } 1 \leq j \leq k_2$.

            We know that $\expsubst{(M_1 M_2)}{x}{N} \singlearrow \expsubst{M_1}{x}{N} \expsubst{M_2}{x}{N}$.

            Now, \\
            $\finexpansion{\expsubst{M_1}{x}{N} \expsubst{M_2}{x}{N} : \sigma}{((\boudolsubst{\transboudol{M_1}}{\finboudol{P}{s}}{x}\finboudol{(\boudolsubst{Q}{\finboudol{P}{(n_1 + \cdots + n_{k_2})-s}}{x})}{m_1 + \cdots + m_{k_1}}), A \otimes B_1 \otimes \cdots \otimes B_{k_2})}$.
            
            By the definition of expansion context, we have that $A(x) = \typemult{(\sigma_1 \cap \cdots \cap \sigma_{k_2})}{n_1 + \cdots + n_{k_2}}$ if and only if, for some $s$, such that $0 \leq s \leq n_1 + \cdots + n_{k_2}$, $\exists \psi_1, \psi_2$, such that $x : \typemult{\psi_1}{s} \in A_0$ and $x : \typemult{\psi_2}{(n_1 + \cdots + n_{k_2})-s} \in (A_1 \otimes \cdots \otimes A_{k_1})$. 

            Consequently, we have \\
            $\finexpansion{\expsubst{M_1}{x}{N}: \typemult{\tau_1}{m_1} \cap \cdots \cap \typemult{\tau_{k_1}}{m_{k_1}} \to \sigma}{(\boudolsubst{\transboudol{M_1}}{\finboudol{P}{s}}{x}, A_0 \otimes C_1)}$, where \\
            $C_1 \subseteq (B_1 \otimes \cdots \otimes B_{k_2})$, because $\finexpansion{M_1 : \typemult{\tau_1}{m_1} \cap \cdots \cap \typemult{\tau_{k_1}}{m_{k_1}} \to \sigma}{(\transboudol{M_1}, A_0)}$ and, since $A_0(x) = \typemult{\psi_1}{s}$ then $N$ is only expanded $s$ times with intersection type $\psi_1$.

            We also have that $\finexpansion{\expsubst{M_2}{x}{N}}{((\boudolsubst{Q_i}{\finboudol{P}{(n_1 + \cdots + n_{k_2})-s}}{x}), A_1 \otimes \cdots \otimes A_{k_1} \otimes C_2)}$, where $C_2 \subseteq (B_1 \otimes \cdots \otimes B_{k_2})$, because $\finexpansion{M_2 : \tau_i}{(Q_i, A_i)}$ and, since $(A_1 \otimes \cdots \otimes A_{k_1})(x) = \typemult{\psi_2}{(n_1 + \cdots + n_{k_2})-s}$ then $N$ is only expanded $(n_1 + \cdots + n_{k_2})-s$ times with intersection type $\psi_2$.

            By \cref{boudol:subst_app_lemma}.\ref{boudol:subst_app_lemma_fin}, we have that \\
            $\boudolsubst{(\transboudol{M_1}\finboudol{Q}{m_1 + \cdots + m_{k_1}})}{\finboudol{P}{n_1 + \cdots + n_{k_2}}}{x} \boudolequal (\boudolsubst{\transboudol{M_1}}{\finboudol{P}{s}}{x}\finboudol{(\boudolsubst{Q}{\finboudol{P}{(n_1 + \cdots + n_{k_2})-s}}{x})}{m_1 + \cdots + m_{k_1}})$. 

        \item Given $\finexpansion{\expsubst{(M_1 M_2)}{x}{N}: \sigma}{(\boudolsubst{(\transboudol{M_1}\finboudol{Q}{m_1 + \cdots + m_{k_1}})}{\finboudol{P}{n_1 + \cdots + n_{k_2}}}{x}, A \otimes B_1 \otimes \cdots \otimes B_{k_2})}$, where $x \in \textit{fv}(M_1)$ and $x \notin \textit{fv}(M_2)$, because
            \[ \finexpansion{M_1M_2: \sigma}{(\transboudol{M_1}\finboudol{Q}{m_1 + \cdots + m_{k_1}}, A_0 \otimes A_1 \otimes \cdots \otimes A_{k_1})} \]
            where, for some $k_1>0$, $\finexpansion{M_1: \typemult{\tau_1}{m_1} \cap \cdots \cap \typemult{\tau_{k_1}}{m_{k_1}} \to \sigma}{(\transboudol{M_1}, A_0)}$ and $\finexpansion{M_2 : \tau_i}{(Q_i, A_i)}$, where $1 \leq i \leq k_1$
            and, for some $k_2 > 0$, $A(x) = \typemult{(\sigma_1 \cap \cdots \cap \sigma_{k_2})}{n_1 + \cdots + n_{k_2}}$, with \\
            $A = A_0 \otimes \cdots \otimes A_{k_1}$. And $\finexpansion{N : \sigma_j}{(P_j, B_j)} \text{, where } 1 \leq j \leq k_2$.

            We have $\boudolsubst{(\transboudol{M_1}\finboudol{Q}{m_1 + \cdots + m_{k_1}})}{\finboudol{P}{n_1 + \cdots + n_{k_2}}}{x} \boudoldblarrow (M_1{'}\finboudol{Q}{m_1 + \cdots + m_{k_1}})$ by rule S$_2$, because \\
            $\boudolsubst{\transboudol{M_1}}{\finboudol{P}{n_1 + \cdots + n_{k_2}}}{x} \boudoldblarrow M_1{'}$.

            We know that $\expsubst{(M_1 M_2)}{x}{N} \singlearrow \expsubst{M_1}{x}{N} M_2$. And \\
            $\finexpansion{\expsubst{M_1}{x}{N} M_2:\sigma}{((\boudolsubst{\transboudol{M_1}}{\finboudol{P}{n_1 + \cdots + n_{k_2}}}{x}\finboudol{Q}{m_1 + \cdots + m_{k_1}}), A \otimes B_1 \otimes \cdots \otimes B_{k_2})}$. By rule E$_1$, \\
            $(\boudolsubst{\transboudol{M_1}}{\finboudol{P}{n_1 + \cdots + n_{k_2}}}{x}\finboudol{Q}{m_1 + \cdots + m_{k_1}}) \boudoldblarrow (M_1{'}\finboudol{Q}{m_1 + \cdots + m_{k_1}})$, because $\boudolsubst{\transboudol{M_1}}{\finboudol{P}{n_1 + \cdots + n_{k_2}}}{x} \boudoldblarrow M_1{'}$.

        \item Given $\finexpansion{\expsubst{(M_1 M_2)}{x}{N}: \sigma}{(\boudolsubst{(\transboudol{M_1}\finboudol{Q}{m_1 + \cdots + m_{k_1}})}{\finboudol{P}{n_1 + \cdots + n_{k_2}}}{x}, A \otimes B_1 \otimes \cdots \otimes B_{k_2})}$, where $x \notin \textit{fv}(M_1)$ and $x \in \textit{fv}(M_2)$, because $\finexpansion{M_1M_2: \sigma}{(\transboudol{M_1}\finboudol{Q}{m_1 + \cdots + m_{k_1}}, A_0 \otimes A_1 \otimes \cdots \otimes A_{k_1})}$ 
            where, for some $k_1>0$, $\finexpansion{M_1: \typemult{\tau_1}{m_1} \cap \cdots \cap \typemult{\tau_{k_1}}{m_{k_1}} \to \sigma}{(\transboudol{M_1}, A_0)}$ and \\
            $\finexpansion{M_2 : \tau_i}{(Q_i, A_i)} \text{, where } 1 \leq i \leq k_1$
            and, for some $k_2 > 0$, $A(x) = \typemult{(\sigma_1 \cap \cdots \cap \sigma_{k_2})}{n_1 + \cdots + n_{k_2}}$, with $A = A_0 \otimes \cdots \otimes A_{k_1}$. And $\finexpansion{N : \sigma_j}{(P_j, B_j)} \text{, where } 1 \leq j \leq k_2$.

            Since $x \notin \textit{fv}(M_1)$, then $\transboudol{M_1} \equiv \lambda y. \transboudol{M_1{'}}$. Therefore, \\
            $\boudolsubst{((\lambda y. \transboudol{M_1{'}})\finboudol{Q}{m_1 + \cdots + m_{k_1}})}{\finboudol{P}{n_1 + \cdots + n_{k_2}}}{x} \boudolarrow \boudolsubst{(\boudolsubst{\transboudol{M_1{'}}}{\finboudol{Q}{m_1 + \cdots + m_{k_1}}}{y})}{\finboudol{P}{n_1 + \cdots + n_{k_2}}}{x}$

            We know that $\expsubst{(M_1 M_2)}{x}{N} \singlearrow M_1 \expsubst{M_2}{x}{N}$. And \\
            $\finexpansion{M_1 \expsubst{M_2}{x}{N}:\sigma}{((\transboudol{M_1}\finboudol{(\boudolsubst{Q}{\finboudol{P}{n_1 + \cdots + n_{k_2}}}{x})}{m_1 + \cdots + m_{k_1}}), A \otimes B_1 \otimes \cdots \otimes B_{k_2})}$.

            Since $\transboudol{M_1} \equiv \lambda y. M_1{'}$, \\
            $((\lambda y. \transboudol{M_1{'}})\finboudol{(\boudolsubst{Q}{\finboudol{P}{n_1 + \cdots + n_{k_2}}}{x})}{m_1 + \cdots + m_{k_1}}) \boudolarrow \boudolsubst{\transboudol{M_1{'}}}{\finboudol{(\boudolsubst{Q}{\finboudol{P}{n_1 + \cdots + n_{k_2}}}{x})}{m_1 + \cdots + m_{k_1}}}{y}$.

            Since $x \notin \textit{fv}(M_1)$, $\boudolsubst{(\boudolsubst{\transboudol{M_1{'}}}{\finboudol{Q}{m_1 + \cdots + m_{k_1}}}{y})}{\finboudol{P}{n_1 + \cdots + n_{k_2}}}{x} \boudolequal \boudolsubst{\transboudol{M_1{'}}}{\finboudol{(\boudolsubst{Q}{\finboudol{P}{n_1 + \cdots + n_{k_2}}}{x})}{m_1 + \cdots + m_{k_1}}}{y}$.
    \end{itemize}
    \item Inductive case: $\finexpansion{MN:\sigma}{((\transboudol{M}\finboudol{Q}{m_1 + \cdots + m_{k_1}}), A_0 \otimes A_1 \otimes \cdots \otimes A_k)}$ and \\ 
    $\finexpansion{(\expsubst{M}{x}{N_1})N_2 : \sigma}{((\boudolsubst{\transboudol{M}}{\finboudol{P}{n_1 + \cdots + n_{k_2}}}{x}\finboudol{Q}{m_1 + \cdots + m_{k_1}}), A_0 \otimes \cdots \otimes A_{k_1} \otimes B_1 \otimes \cdots B_{k_2})}$ follow easily by induction.
\end{itemize}

\end{proof}

\begin{example} \label[example]{exp:theo_exp_fin_example}
    Let $I \equiv \lambda y.y$, $\Delta \equiv \lambda z.zz$ and $\sigma \equiv \alpha \to \alpha$.

    We know that $\expsubst{(\expsubst{(fx)}{x}{I})}{f}{\Delta} \singlearrow \expsubst{(f\expsubst{x}{x}{I})}{f}{\Delta}$ and, from \\
    $\finexpansion{f:((\sigma \to \sigma) \cap \sigma) \to \sigma}{(f, \{f : \typemult{((\sigma \to \sigma) \cap \sigma) \to \sigma}{1}\})}$, \\
    $\finexpansion{x : \sigma \to \sigma}{(x, \{x : \typemult{(\sigma \to \sigma)}{1}\})}$, $\finexpansion{x:\sigma}{(x, \{x : \typemult{\sigma}{1}\})}$, we have \\
    $\finexpansion{fx : \sigma}{((f\finboudol{x}{2}), \{f:\typemult{((\sigma \to \sigma) \cap \sigma) \to \sigma}{1}, x: \typemult{((\sigma \to \sigma) \cap \sigma)}{2}\})}$, and from \\
    $\finexpansion{I : \sigma \to \sigma}{I}$ and $\finexpansion{I:\sigma}{I}$, we have \\
    $\finexpansion{\expsubst{(fx)}{x}{I}:\sigma}{(\boudolsubst{(f\finboudol{x}{2})}{\finboudol{I}{2}}{x}, \{f:\typemult{((\sigma \to \sigma) \cap \sigma) \to \sigma}{1}, x: \typemult{((\sigma \to \sigma) \cap \sigma)}{2}\})}$, and from \\
    $\finexpansion{z : \sigma \to \sigma}{(z, \{z : \typemult{(\sigma \to \sigma)}{1}\})}$ and $\finexpansion{z : \sigma}{(z, \{z : \typemult{\sigma}{1}\})}$, we get \\
    $\finexpansion{zz:\sigma}{((z\finboudol{z}{1}), \{z : \typemult{((\sigma \to \sigma) \cap \sigma)}{1}\})}$, and from $\finexpansion{\Delta:((\sigma \to \sigma) \cap \sigma) \to \sigma}{\Delta_1}$, where \\
    $\Delta_1 \equiv \lambda z. (z\finboudol{z}{1})$, we get \\
    $\finexpansion{\expsubst{(\expsubst{(fx)}{x}{I})}{f}{\Delta}:\sigma}{(\boudolsubst{(\boudolsubst{(f\finboudol{x}{2})}{\finboudol{I}{2}}{x})}{\Delta_1}{f}, \{f:\typemult{(((\sigma \to \sigma) \cap \sigma) \to \sigma)}{1}, x: \typemult{((\sigma \to \sigma) \cap \sigma)}{2}\})}$    
 
    \[
    \begin{array}{rcll}
        \boudolsubst{(\boudolsubst{(f\finboudol{x}{2})}{\finboudol{I}{2}}{x})}{\Delta_1}{f} \boudolarrow \boudolsubst{(\boudolsubst{(\Delta_1\finboudol{x}{2})}{1}{f})}{\finboudol{I}{2}}{x} &\boudolarrow& \boudolsubst{(\boudolsubst{(\boudolsubst{(z\finboudol{z}{1})}{\finboudol{x}{2}}{z})}{1}{f})}{\finboudol{I}{2}}{x} \\
        &\boudolarrow& \boudolsubst{(\boudolsubst{(\boudolsubst{(x\finboudol{z}{1})}{x}{z})}{1}{f})}{\finboudol{I}{2}}{x} \\
        &\boudolarrow& \boudolsubst{(\boudolsubst{(\boudolsubst{(I\finboudol{z}{1})}{I}{x})}{x}{z})}{1}{f} \\
        &\boudolarrow& \boudolsubst{(\boudolsubst{(\boudolsubst{(\boudolsubst{y}{z}{y})}{I}{x})}{x}{z})}{1}{f} \\
        &\boudolarrow& \boudolsubst{(\boudolsubst{(\boudolsubst{(\boudolsubst{z}{1}{y})}{I}{x})}{x}{z})}{1}{f} \\
        &\boudolarrow& \boudolsubst{(\boudolsubst{(\boudolsubst{(\boudolsubst{x}{1}{z})}{1}{y})}{I}{x})}{1}{f} \\
        &\boudolarrow& \boudolsubst{(\boudolsubst{(\boudolsubst{(\boudolsubst{I}{1}{x})}{1}{z})}{1}{y})}{1}{f} \equiv I
    \end{array}
    \]
    Now, from $\finexpansion{f:((\sigma \to \sigma) \cap \sigma) \to \sigma}{(f, \{f : \typemult{((\sigma \to \sigma) \cap \sigma) \to \sigma}{1}\})}$, and \\
    $\finexpansion{x : \sigma \to \sigma}{(x, \{x : \typemult{(\sigma \to \sigma)}{1}\})}$, $\finexpansion{x:\sigma}{(x, \{x : \typemult{\sigma}{1}\})}$, 
    $\finexpansion{I : \sigma \to \sigma}{I}$ and $\finexpansion{I:\sigma}{I}$, we have $\finexpansion{f\expsubst{x}{x}{I}:\sigma}{((f\finboudol{(\boudolsubst{x}{I}{x})}{2}), \{f:\typemult{((\sigma \to \sigma) \cap \sigma) \to \sigma}{1}, x: \typemult{((\sigma \to \sigma) \cap \sigma)}{2}\})}$, and from $\finexpansion{z : \sigma \to \sigma}{(z, \{z : \typemult{(\sigma \to \sigma)}{1}\})}$ and $\finexpansion{z : \sigma}{(z, \{z : \typemult{\sigma}{1}\})}$, we get \\
    $\finexpansion{zz:\sigma}{((z\finboudol{z}{1}), \{z : \typemult{((\sigma \to \sigma) \cap \sigma)}{1}\})}$, and from $\finexpansion{\Delta:((\sigma \to \sigma) \cap \sigma) \to \sigma}{\Delta_1}$, where $\Delta_1 \equiv \lambda z. (z\finboudol{z}{1})$, we get \\
    \[
    \begin{array}{rcll}
         &&\finexpansion{\expsubst{(f(\expsubst{x}{x}{I}))}{f}{\Delta}:\sigma}{(\boudolsubst{(f\finboudol{(\boudolsubst{x}{I}{x})}{2})}{\Delta_1}{f}, \\
    &&\hspace{5.4cm}\{f:\typemult{(((\sigma \to \sigma) \cap \sigma) \to \sigma)}{1}, x: \typemult{((\sigma \to \sigma) \cap \sigma)}{2}\})}
    \end{array}
    \]
    \[
    \begin{array}{rcll}
        (\boudolsubst{(f\finboudol{(\boudolsubst{x}{I}{x})}{2})}{\Delta_1}{f} \boudolarrow (\boudolsubst{(\Delta_1\finboudol{(\boudolsubst{x}{I}{x})}{2})}{1}{f} &\boudolarrow& \boudolsubst{(\boudolsubst{(z\finboudol{z}{1})}{\finboudol{(\boudolsubst{x}{I}{x})}{2}}{z})}{1}{f} \\
        &\boudolarrow& \boudolsubst{(\boudolsubst{(\boudolsubst{x}{I}{x}\finboudol{z}{1})}{\boudolsubst{x}{I}{x}}{z})}{1}{f} \\
        &\boudolarrow& \boudolsubst{(\boudolsubst{(\boudolsubst{I}{1}{x}\finboudol{z}{1})}{\boudolsubst{x}{I}{x}}{z})}{1}{f} \\
        &\boudolarrow& \boudolsubst{(\boudolsubst{(\boudolsubst{(\boudolsubst{y}{z}{y})}{1}{x})}{\boudolsubst{x}{I}{x}}{z})}{1}{f} \\
        &\boudolarrow& \boudolsubst{(\boudolsubst{(\boudolsubst{(\boudolsubst{z}{1}{y})}{1}{x})}{\boudolsubst{x}{I}{x}}{z})}{1}{f} \\
        &\boudolarrow& \boudolsubst{(\boudolsubst{(\boudolsubst{(\boudolsubst{(\boudolsubst{x}{I}{x})}{1}{z})}{1}{y})}{1}{x})}{1}{f} \\
        &\boudolarrow& \boudolsubst{(\boudolsubst{(\boudolsubst{(\boudolsubst{(\boudolsubst{I}{1}{x})}{1}{z})}{1}{y})}{1}{x})}{1}{f} \equiv I
    \end{array}
    \]
\end{example}

\begin{corollary} \label[corollary]{exp:cor_exp_fin}
    Given a $\lambda$x-term $M$ and an AC-intersection type $\sigma$, such that $\finexpansion{M: \sigma}{(\transboudol{M}, A)}$, if $M \dblarrow V$ then $\transboudol{M} \boudolequal V'$ and $\finexpansion{V:\sigma}{(V', B)}$, where $V$ and $V'$ are weak-head normal forms and $B \subseteq A$.
\end{corollary}
\begin{example} \label[example]{exp:cor_exp_fin_example}
    Following the example in \cref{exp:theo_exp_fin_example}, we have 
    \[
    \begin{array}{rcll}
        \expsubst{(\expsubst{(fx)}{x}{I})}{f}{\Delta} &\singlearrow& \expsubst{(f(\expsubst{x}{x}{I}))}{f}{\Delta} \\
        &\singlearrow& \expsubst{f}{f}{\Delta}(\expsubst{x}{x}{I}) \\
        &\singlearrow& \Delta(\expsubst{x}{x}{I}) \\
        &\singlearrow& \expsubst{(zz)}{z}{(\expsubst{x}{x}{I})} \\
        &\singlearrow& \expsubst{z}{z}{(\expsubst{x}{x}{I})}(\expsubst{z}{z}{(\expsubst{x}{x}{I})}) \\
        &\singlearrow& \expsubst{x}{x}{I}(\expsubst{z}{z}{(\expsubst{x}{x}{I})}) \\
        &\singlearrow& I\expsubst{z}{z}{(\expsubst{x}{x}{I})} \\
        &\singlearrow& \expsubst{y}{y}{(\expsubst{z}{z}{(\expsubst{x}{x}{I})})} \\
        &\singlearrow& \expsubst{z}{z}{(\expsubst{x}{x}{I})} \singlearrow \expsubst{x}{x}{I} \singlearrow I 
    \end{array}
    \] 
    We already know that $\boudolsubst{(\boudolsubst{(f\finboudol{x}{2})}{\finboudol{I}{2}}{x})}{\Delta_1}{f} \boudoldblarrow I$ and $\finexpansion{I:\sigma}{I}$.
\end{example}

Note that the expansion relation heavily depends on the intersection type used to expand a term. As an example, if one tries to expand the term $\lambda x.xx$ with type $\alpha \cap \alpha \to \alpha$ one can easily see that this is impossible because we are not able to expand $xx$ using type $\alpha$, using $x : \alpha \cap \alpha$ in the expansion context. In fact, by definition, in the expansion of the first occurence of $x$, the type must be of the form $\tau_1 \cap \cdots \cap \tau_k \to \sigma$, but it is $\alpha$ in this example. 

\subsection{Expansion and weak-head reduction}
Weak-head reduction is crucial in the definition of a sound term expansion. In fact, expansion is preserved by weak-head reduction, as demonstrated in \cref{exp:theo_exp_inf} and \cref{exp:theo_exp_fin}. But note that the same does not happen for $\lambda$xgc-reduction as originally stated in \cref{expsubst:reduction}. Consider the following example.
\begin{example} \label[example]{exp:weak_head_red_example}
    \[
        \lambda x. (\lambda y. z) xx \singlearrow \lambda x. \expsubst{z}{y}{x} x \singlearrow \lambda x. zx
    \]

    \[ \infexpansion{\lambda x. (\lambda y. z) xx : \alpha_1 \cap \alpha_2 \to \beta}{\lambda x. ((\lambda y. z)\infboudol{(x\infboudol{x})})} \]

    $\lambda x. (\lambda y. z)\infboudol{(x\infboudol{x})} \boudoldblarrow \lambda x. ((\lambda y. z)\infboudol{(x\infboudol{x})})$
\end{example}
By \cref{exp:theo_exp_inf} we can easily see that this does not hold because \\
$\infexpansion{\lambda x. zx : \alpha_1 \cap \alpha_ 2 \to \beta_1}{\lambda x. (z \infboudol{x})}$ and $\lambda x. ((\lambda y. z)\infboudol{(x\infboudol{x})}) \not\boudolequal \lambda x. (z \infboudol{x})$.
\\
\\
Note that the original definition of expansion \cite{florido2004linearization} also preserves weak-head reduction and does not preserve strong reduction. This is essentially due to the lack of subject reduction of the intersection type system, implicitly used to direct the expansion process.
 \section{Conclusions}
In this paper, we proved that there exists a relation between associative, commutative and idempotent intersection types and a resource calculus that deals with infinite multiplicities, $m = \infty$, and we also proved that this relation may be extended to the use of associative, commutative and non-idempotent intersection types which, in this case, defines a relation with a resource calculus of finite multiplicities, $m \in \mathbb{N}$, where conjunction is managed in a multiplicative manner, according to linear logic terminology. 

In the future, we want to study the application of {\em term expansion} to relate the $\lambda$-calculus with concurrent calculi. Note that concurrent calculi \cite{milner1992calculus} are more discriminating than the $\lambda$-calculus, in the sense that they have a limited availability of resources. The relation between the $\lambda$-calculus with explicit substitutions, of finite or infinite multiplicities, and the $\pi$-calculus is not new \cite{boudol1995lambda}, and this previous work strongly suggests that expansion may have a role in establishing new relations between the two calculi.
\newline
\newline
{\bf Acknowledgments:} This work was financially supported by: UID/00027/2025 of the LIACC - Artificial Intelligence and Computer Science Laboratory with DOI https://doi.org/10.54499/UID/00027/2025 funded by Fundação para a Ciência e a Tecnologia, I.P./ MECI through the national funds. 
Ana Jorge Almeida is supported by Fundação para a Ciência e a Tecnologia (Portuguese Foundation for Science and Technology) through the Carnegie Mellon Portugal Program under the fellowship 
2025.15429.PRT. 
Sandra Alves is supported by national funds through FCT – Fundação para a Ciência e a Tecnologia, I.P., under the support UID/50014/2025 (https://doi.org/10.54499/UID/50014/2025).

\bibliographystyle{eptcs}
\bibliography{refs}
\end{document}